%% file: cosmo.tex
\documentclass[reqno]{amsart}
\input{definitions.tex}
\gettimestamp Time-stamp: <2020-05-07 15:35:00 CDT>

\begin{document}

\title[Horseshoes in cosmological models]{Horseshoes and invariant
  tori in cosmological models with a coupled field and non-zero curvature}
\author{Leo T. Butler}
\address{Department of Mathematics, University of Manitoba, Winnipeg,
  MB, Canada, R2J 2N2}
\email{leo.butler@umanitoba.ca}
\date{\timestamp}
\thanks{Partially supported by the Natural Science and
  Engineering Research Council of Canada grant 320 852.}

\begin{arxivabstract}
  This paper studies the dynamics of a family of hamiltonian systems
  that originate from Friedman-Lema{î}tre-Robertson-Walker
  space-times with a coupled field and non-zero curvature. In four
  distinct cases, previously considered by Maciejewski, Przybylska,
  Stachowiak \& Szydowski, it is shown that there are homoclinic
  connections to invariant submanifolds and the connections
  split. These results imply the non-existence of a real-analytic
  integral independent of the hamiltonian.
\end{arxivabstract}

\maketitle

\section{Introduction}
\label{sec:intro}

Almost one hundred years ago, in a pair of ground-breaking papers,
Alexander Friedman introduced a simplified solution to Einstein's
equations of general
relativity~\cite{Friedman1922,springer_jour1026751225741,Friedman1924,alma99144120760001651}. This
solution implicitly determines the evolution of the radius of the
universe~\cite[Equation 5]{Friedman1922}. Lema{î}tre independently
rediscovered this equation~\cite[Equation
2]{Lemaitre1927}. Robertson~\cite{1935ApJ....82..284R,1936ApJ....83..187R,1936ApJ....83..257R}
and Walker~\cite{1937PLMS...42...90W} separately considered spatially
homogeneous space-times and their properties.

In~{\cite{MR2515846}}, \Maciejewski{} investigate two families of
hamiltonians motivated by the \FLRW{} models in cosmology. The first,
the \defn{minimally-coupled field}, has the
hamiltonian~\cite[eq. 10]{MR2515846}
\begin{equation}
  \label{eq:maciej-min-couple-H}
  H = - ½ \left( A^2 + 2 ka^2 - 2 Λ a^4 \right) + ½ \left( a^{-2} B^2 + 2 (ω/ab)^2 + 2 (ma^2b)^2 \right),
\end{equation}
where $k, Λ$ and $m$ are scalar parameters and $ω$ is angular
momentum. As noted, \Maciejewski{} consider only the case where
angular momentum vanishes. They prove in Theorem 5.i, that if $H$ is
integrable, then $9 - 4 m^2/Λ$ is a perfect square. Conjecture 5.1.i
conjectures that, in fact, the only integrable case is $m=0$. The
first result of the present paper is

\begin{theorem}
  \label{thm:proof-of-conj-5.1.i}
  Assume $k, Λ > 0$ and $ω$ are fixed. For all $ω$ sufficiently small,
  if the hamiltonian \eqref{eq:maciej-min-couple-H} has a second,
  independent real-analytic integral of motion, then $m=0$.
\end{theorem}

This theorem is proven by demonstrating the existence of a horseshoe
in the dynamics of a family of hamiltonians that includes the
unreduced variant of \eqref{eq:maciej-min-couple-H}.

In the same paper, the authors consider a second hamiltonian derived
from a \defn{minimally-conformally-coupled field}~\cite[eq. 18]{MR2515846}
\begin{equation}
  \label{eq:maciej-conf-couple-H}
  H = - ½ \left( A^2 + k a^2 - Λ a^4 \right) + ½ \left( B^2 + k b^2 + (ω/b)^2 + ½ λ b^4 + (mab)^2 \right)
\end{equation}
where the parameters are the same as in \eqref{eq:maciej-min-couple-H}
with the exception that a self-excitation term has been added with a
strength of $λ$. In Theorem 7.i, it is proven that if the hamiltonian
is integrable and angular momentum vanishes, then either either $k=0$
or $λ=Λ$ and $m^2=-Λ,-3 Λ$.

\begin{theorem}
  \label{thm:proof-of-conf-couple}
  Assume that $k, Λ > 0$ and $λ, ω$ are fixed. For all $ω$
  sufficiently small, if the
  hamiltonian~\eqref{eq:maciej-conf-couple-H} has a second,
  independent real-analytic integral of motion, then $m$ is imaginary.
\end{theorem}

Similar to Theorem~\ref{thm:proof-of-conj-5.1.i}, this theorem is
proven by demonstrating the existence of horseshoes in the dynamics;
perhaps surprisingly, the proof also suggests the integrable case
where $m^2=-Λ$, but makes no constraint on $λ$.

The technique used to prove these theorems exploits a well-known
mechanism: both hamiltonians enjoy saddle-centre equilibria with
connecting orbits. In a suitably re-scaled limit, the DE decouple and
enjoy normally hyperbolic invariant manifolds which are foliated by
invariant tori. Using essentially the same machinery as the
variational DE employed in~\cite{MR2515846}, one computes the
\poincaremelnikov{} function $M$ for the connecting orbits and shows
that $M$ has non-degenerate zeros. By well-known arguments, this
proves the existence of transverse homo/hetero-clinic orbits, a
suspended horseshoe and the absence of a second constant of motion
that is independent of the Hamiltonian $H$.

The case of zero curvature and angular momentum ($k=0=ω$) is studied
in~\cite{MR2790075} by Mahdi, Llibre and Valls. They use the weighted
homogeneity of the hamiltonian to prove that there does not exist a
second real-analytic constant of motion except in the known integrable
cases (see below). Although their results are only stated for the case
where the kinetic part of the hamiltonian is positive definite, the
arguments based on \cite[Proposition 2]{MR2790075}, extend to the
indefinite case. It should be noted that this result is weaker than
theorems~\ref{thm:proof-of-conj-5.1.i} and
\ref{thm:proof-of-conf-couple}: it remains an open question if there
are horseshoes in the dynamics.

In the case of negative curvature ($k=-1$), the present results are
less definitive than the positive curvature case. It is proven that

\begin{theorem}
  \label{thm:proof-of-conj-5.1.i-k=-1}
  Assume that $k, Λ < 0$ and $ω$ are fixed. For all $ω$ sufficiently
  small, if the hamiltonian~\eqref{eq:maciej-min-couple-H} has a
  second, independent real-analytic integral of motion, then $m=0$.
\end{theorem}

and

\begin{theorem}
  \label{thm:proof-of-conf-couple-k=-1}
  Assume that $k, Λ, -λ < 0$ and $ω$ are fixed. For all $ω$
  sufficiently small there is a countable set $E_{ω}$ of real numbers
  such that, if the
  hamiltonian~\eqref{eq:maciej-conf-couple-H} has a second,
  independent real-analytic integral of motion, then $m ∈ E_{ω}$.
\end{theorem}

The last theorem is proven using similar tools to the first three
theorems, but the details are rather different as is the result. The
origin is a saddle critical point for the unreduced hamiltonian (for
all $m$ and $ω=0$). With the assumption that $λ > 0$, when $m=0$ the
saddle's stable and unstable manifolds coincide. It is shown that for
$m ≠ 0$ sufficiently small, these manifolds split and create
transverse homoclinic orbits. This implies that a family of nearby
hyperbolic periodic orbits also have transverse homoclinic orbits and
hence horseshoes in the dynamics. The homoclinic orbits where the
splitting is detected exist for all $m$ and due to the real-analytic
dependence of the stable and unstable manifolds of the saddle, the
set $E_0$, of $m$ where those manifolds are not transverse along the
homoclinic orbits, is a closed real-analytic subset of the reals with
a non-empty complement. By real-analyticity in $ω$, there is a
similarly defined set $E_{ω}$ for each $ω$ sufficiently small. It is
likely that $E_{ω}$ is empty in all cases, but the present techniques
cannot prove this. On the other hand,
theorem~\ref{thm:proof-of-conf-couple-k=-1} is the only theorem where
we show that the horseshoe is on the zero energy level.

\subsection{Outline}
\label{sec:outline}

The outline of the present note is: \S~\ref{sec:non-int-chaos} reviews
related work; \S~\ref{sec:frw-spt} reviews the Lagrangian derivation
of the hamiltonians following the presentation in~\cite{MR2515846};
\S~\ref{sec:minimal-coupling} sets up the \poincaremelnikov{} integral
for the saddle connections of~\eqref{eq:maciej-min-couple-H};
\S~\ref{sec:min-conf-couple} does likewise
for~\eqref{eq:maciej-conf-couple-H}; \S~\ref{sec:computation} explains
the computation of the \poincaremelnikov{} integral in 3 of the 4
cases (the remaining case is dealt with
in~\S~\ref{sec:min-conf-couple-k=-1}); \S~\ref{sec:kam-tori} proves
the existence of KAM tori in the case of $k=-1$ with minimal conformal
coupling; \S~\ref{sec:ack} has acknowledgments;
\S~\ref{sec:conclusion} concludes; Figures and references follow.

\section{Non-integrable and chaotic dynamics}
\label{sec:non-int-chaos}

The existence of non-integrable or chaotic dynamics in several
cosmological models is well-known. Belinsky, Khalatnikov \&
Lifshitz~\cite{informaworld_s10_1080_00018737000101171} conjecture
that the nature of singularities in space-time are dictated
asymptotically by Bianchi IX space-time and that the transition
between singularities is governed by the Gauss map. This work was
later amplified in~\cite{springer_jour10.1007/BF01017851}. It should
be noted that this characterization remains conjectural--Cushman \&
\'{S}niatycki prove that the hamiltonian flow is locally integrable
but the proof is not constructive~\cite{MR1377227}. Indeed, to prove
integrability, they construct a Lyapunov function (which by flow-box
coordinates yields local integrability), but the flow has no
recurrence. In~\cite{arxivgr-qc/9711014}, de Oliveira, Soares and
Stuchi demonstrate chaotic dynamics in a reduction of the Bianchi IX
model coupled with a scalar field.

In the case of \FLRW{} models, Calzetta and El Hasi study the minimal
conformally coupled model with a real scalar field and $λ=Λ=0$,
i.e. there is no accelerating inflation nor self-excitation of the
scalar field~\cite[eq. 5]{Calzetta_1993}. They demonstrate the
existence of horseshoes on the zero energy level and provide numerical
phase portraits as evidence of it, too. In a related vein, Bombelli
and Calzetta show that a relativistic particle in motion around a
Schwarzschild black-hole has hyperbolic periodic orbits with coincident
homoclinic connections; for a generic periodic perturbation of the
Schwarzschild metric, the connections split and create
horseshoes~\cite{Bombelli_1992}. Bombelli, Lombardo and Castagnino
revisit the work of Calzetta and El Hasi and expand upon the
computation of the \poincaremelnikov{} integrals in the former
paper~\cite{aip_complete10.1063/1.532612}. De Oliveira \& Soares
consider a \textem{frustrated} variant of the minimally-coupled
hamiltonian~\eqref{eq:maciej-min-couple-H} with a real-scalar field:
frustrated in this case means their hamiltonian is $F=-(H-E_0)/a$
where $E_0$ is the energy of the saddle-centre equilibrium, so for
$F ≠ 0$ the sign of $a$ is frustrated from
changing~\cite[eq. 1]{MR1712558}. They present numerical evidence and
offer a heuristic reason that in a neighbourhood of the saddle-centre
equilibrium there is a family of hyperbolic periodic orbits with split
homoclinic connections and therefore horseshoes. Other work connects
this with a possible mechanism to explain
inflation~\cite{MR2206428,arxivgr-qc/9711014}.

Ziglin's ground-breaking work on meromorphic integrability included
his proof that the Yang-Mills hamiltonian is
non-integrable~\cite{MR674006,MR695092}. The Yang-Mills hamiltonian
can be obtained from the hamiltonian~\eqref{eq:maciej-conf-couple-H}
by specializing $λ=Λ=ω=k=0$ and applying a complex rotation in the
$(a,A)$ plane--such a change of variables destroys the \textem{real}
phase portrait of the hamiltonian but it leaves invariant its
integrability in the class of meromorphic
integrals. \Coelhoetal{}~\cite{Coelho_2008} trod similar ground to
\Maciejewski{}: they use differential Galois theory as developed by
Morales-Ruiz and Ramis to demonstrate the non-integrability of the
minimal conformally coupled
hamiltonian~\eqref{eq:maciej-conf-couple-H} when $ω=0$. They show that
when $Λ=λ=0$ and $m ≠ 0$, then the hamiltonian does not possess a
second, independent meromorphic constant of motion and when
$k,λ,Λ ≠ 0$, the same is true except when $Λ=λ=-m^2$ or
$-m^2/3$~\cite[Theorems 3,5]{Coelho_2008}. It should be noted that
when $m ≠ 0$ is real, the first result of \Coelhoetal{} is implied by
the works of Calzetta \& El Hasi and Bombelli, Lombardo \&
Castagnino. Although the latter works seem to imply otherwise
(c.f.~\cite[p. 1828]{Calzetta_1993},
\cite[p. 6048]{aip_complete10.1063/1.532612}), a straightforward
rescaling shows that in this case all such hamiltonian flows are
conjugate up to a constant reparameterization of conformal
time~(c.f.~\cite[eq. 5]{MR2814709}). Helmi \& Vucetich use Painlev{é}
analysis to determine the possible integrable cases
of~\eqref{eq:maciej-conf-couple-H}~\cite{elsevier_sdoi_10_1016_S0375_9601_97_00258_2}.

More recently, Shi \& Li~\cite{MR2995865} examine the generalized Yang-Mills
hamiltonian
\begin{equation}
  \label{eq:gen-ym}
  H = ½ \left( A ² + α a ² \right) + ½ \left( B ² + β b ² \right) + ¼ a ⁴ + ½ μ (ab)² + ¼ η b ⁴,
\end{equation}
from the viewpoint of the theory of Morales-Ruiz \& Ramis and the
higher-order theory of Morales-Ruiz, Ramis \&
Simo~\cite{MR1867495,MR1713573,MR2419851}. There are several known
integrable cases of~\eqref{eq:gen-ym}:
\begin{enumerate}
\item $α=β$, $μ=η=1$: the rotationally-invariant case with $F=aB-Ba$;
\item $μ=0$: the separable case;
\item[(2a)] $α=β$, $μ=3, η=1$ due to Bountis, Segur \&
  Vivaldi~\cite{MR647746}; as noted in~\cite[p. 2293]{MR715400}, this
  a special case of the previous case where $η=1$ and the potential
  separates after a rotation by $π/4$;
\item $β=4 α$, $μ=3, η=8$: Dorizzi, Grammaticos \& Ramani discovered
  this case in their work on Darboux's ``direct method'' for finding
  integrable 2-dimensional potentials. Under suitable simplifying
  assumptions, the potential satisfies a linear second-order PDE; the
  current case in the notation of~\cite[eq. 18]{MR715399} is
  $½ α V_2+¼ V_4$ \cite{MR1353751};
\item $β=4 α$, $μ=6, η=16$: Similar to the previous
  case~\cite[eq. 4.6]{MR715400}, this is a superposition of two
  integrable potentials;
\item $β ≠ α$, $μ=η=1$: This is the 2-dimensional Garnier system
  studied in~\cite{MR1283227}. The integral
  in~\cite[p. 1646]{MR2995865} is incorrect, as is that
  in~\cite[p. 158]{MR1283227}, the correct integral appears on p. 168
  of Vanhaecke's paper.
\end{enumerate}
Shi \& Li demonstrate that when $α ≠ β$, the generalized Yang-Mills
hamiltonian is not meromorphically integrable except for the above
listed cases (3--5). Of the known integrable cases of
\eqref{eq:gen-ym}, only the first two cases are relevant for the
purposes of this paper.


In a sequence of papers, Llibre \& Vidal~\cite{MR2919532}, Lembarki \&
Llibre~\cite{ctx890184400670001651} and Jim{é}nez-Lara \&
Llibre~\cite{MR2814709} use averaging theory to show the existence of
a family of isolated periodic orbits that are parameterized by energy
and have non-trivial Floquet multipliers to the origin in a
hamiltonian motivated by the \FLRW{} model
(c.f. \eqref{eq:maciej-conf-couple-H} and~\eqref{eq:con-couple}
below). The existence of such periodic orbits is taken as an
indication that the hamiltonians do not enjoy a second, independent
$C^1$ first integral. However, two important qualifications need to be
made: first, these arguments can only prove that the hamiltonian
vector field of any first integral must be co-linear along these
orbits to the given vector field--to obtain stronger results, one
needs a topological or metric characterization of the set of such
periodic orbits (e.g. their closure forms a horseshoe); second,
\cite{MR2919532} considers only $k=1$ and $λ, Λ < 0$ and inspection of
\cite[eq. 10]{MR2919532} shows that the proof does not extend to the
case where either $λ > 0$ or $Λ > 0$.

dos Santos \& Vidal consider the stability of the origin for a $2$ and
$3$ degree-of-freedom version of the
hamiltonian~\eqref{eq:maciej-conf-couple-H} for $k=1$~\cite[\S
7]{alma99141287550001651}. They prove for the $2$ degree-of-freedom
case that when $3 Λ + m^2 < 0$ (resp. $>0$) the origin is Lyapunov
unstable (resp. formally stable); and a similar result is proven for
$3$ degrees of freedom with a sparse coupling.

In~\cite{iop10.1088/1361-6544/ab1bc6}, Palacián, Vidal, Vidarte \&
Yanguas study a $3$ degree-of-freedom version of the
hamiltonian~\eqref{eq:maciej-conf-couple-H} for $k=1$ distinct from
the one in the previous paragraph. Similarly, however, the paper
focuses on the critical point at the origin and uses multi-scale KAM
theory to prove the existence of invariant $3$ tori near that critical
point.

\section{\FLRW{} space-time}
\label{sec:frw-spt}

Let us motivate the equations following the approach taken in
\cite{MR2515846}. The metric on space-time, modeled as $\R × M$, is
postulated to be
\begin{equation}
  \label{eq:frw-metric}
  \D{s}^2 = a(η)^2 \, \left( -\D{η}^2 + g \right)
\end{equation}
where $g$ is a metric on the space-like manifold $M$. The time-like
variable $η$ is conformal time; the time measured by an external
observer would be determined by $\D{t} = |a(η)| \D{η}$. The selection
principle for $\D{s}^2$ is determined by the action functional
\begin{equation}
  \label{eq:frw-action}
  I = \int\limits_{\R × M} \left[ \scalarcurvature{} - 2 Λ - ½ \left( \gnormof{∇ Ψ}^2 + V(Ψ) + ξ\, \scalarcurvature{}\, \normof{Ψ}^2  \right) -\rho \right]\, \D{vol}_s
\end{equation}
where $\scalarcurvature{}$ is the Ricci (scalar) curvature of
$\D{s}^2$, $Λ$ is the cosmological constant, $Ψ : \R × M → \R^n$ is a
field, $\normof{Ψ}$ is the euclidean norm, $∇$ is the gradient
operator of the metric $\D{s}^2$ which is extended component-wise for
vector-valued functions, $\gnormof{∇ Ψ}^2$ is the $\D{s}^2$-inner
product of $∇ Ψ$ with itself, also extended component-wise,
$V : \R^n → \R$ is a potential function, $ξ$ is a coupling constant,
$ρ$ is ``fluid'' density and $\D{vol}_s$ is the volume form of
$\D{s}^2$.

Let us assume the following:
\begin{hypotheses}
\item\label{it:a0} $(M,g)$ is a finite-volume homogeneous Riemannian
  manifold whose (constant) scalar curvature is $6k ≠ 0$;
\item\label{it:a00} the volume of $(M,g)$ is unity;
\item\label{it:a1} the field $Ψ$ is spatially homogeneous, hence
  depends only on $η$;
\item\label{it:a2} the density $ρ = c a^{-d}$ where $d=1+\dim M$ and
  $c$ is a constant;
\item\label{it:a3} the potential $V : \R^n → \R$ is a polynomial of
  degree $≤ d$ such that $V$ decomposes into a sum $V_2+ \cdots + V_d$
  where $V_k$ is homogeneous of degree $k$ and $V_2$ is positive
  definite;
\item\label{it:a4} the dimension of space-time $d=4$;
\item\label{it:a5} the cosmological constant $Λ$ has $k Λ > 0$.
\end{hypotheses}

Since these assumptions imply that the integrand of $I$ (Lagrangian)
is independent of the spatial variables, in the case $d=4$ scalar
curvature reduces to $a^4 \scalarcurvature{} = 6 a a'' + 6k a^2$ and
the action functional reduces to
\begin{equation}
  \label{eq:frw-action-reduced}
  I = \int\limits_{\R} \left[ 6(1 - ½ ξ\, \normof{Ψ}^2)(aa'' + k a^2) + ½ \normof{Ψ'}^2 a^2 - a^4 V(Ψ) - 2 Λ a^4 -  c \right]\, \D{η}.
\end{equation}
Integration by parts, combined with the assumption that $a'a$ and $a'a
\normof{Ψ}^2$ are equal at $η = ± ∞$ yields the Lagrangian
\begin{equation}
  \label{eq:frw-lagrangian}
  L = (-6 + 3 ξ \normof{Ψ}^2) (a')^2 + 6 ξ \ip{a Ψ}{a' Ψ'} + ½ a^2 \normof{Ψ'}^2 - ½ a^4 V(Ψ) - 3 ξ k a^2 \normof{Ψ}^2 + 6 k a^2 - 2 Λ a^4 - c,
\end{equation}
where $\ip{}{}$ is the euclidean inner product on $\R^n$. The
``kinetic'' part of the Lagrangian retains an indefinite character for
all $ξ$, but the off-diagonal part makes analysis difficult.

There are two straightforward routes to simplify the Lagrangian further:
\begin{althypotheses}
\item\label{it:min-couple} minimal coupling: set $ξ=0$ to uncouple the field $Ψ$ from the scalar curvature term;
\item\label{it:con-couple} minimal conformal coupling: set $Ψ = τ/a$ and $ξ = 1/6$ to minimize the coupling of the rescaled field.
\end{althypotheses}

\subsubsection{Minimal coupling}
\label{sec:min-couple}

In the case of minimal coupling, the Lagrangian $L$ produces a
hamiltonian $H$, that after suitable rescaling becomes
\begin{equation}
  \label{eq:min-couple}
  H = -\underbrace{½ \left[ A^2 + k a^2 - ½ Λ a^4 \right]}_{H^{(1)}} + \underbrace{½ \left[ a^{-2} \normof{B}^2 + a^4 V(b) \right]}_{H^{(2)}}.
\end{equation}
The hamiltonian $H^{(2)}$ has an apparent singularity at $a=0$;
however, the singularity is not essential and part of the proof below
involves removing the singularity.

\subsubsection{Minimal conformal coupling}
\label{sec:con-couple}

In the case of conformal coupling, after the change of variables an
off-diagonal term is left unless the coupling constant $ξ=1/6$. In
that case the Lagrangian produces a hamiltonian $H$
\begin{equation}
  \label{eq:con-couple}
  H = -\underbrace{½ \left[ A^2 + k a^2 - ½ Λ a^4 \right]}_{H^{(1)}} + \underbrace{½ \left[ \normof{B}^2 + k\normof{b}^2 + a^4 V(b/a) \right]}_{H^{(2)}}.
\end{equation}
In this case, the apparent singularity in $V$ is resolved by the
assumption~\ref{it:a3} on $V$.

\begin{remark}
  \label{rem:flrw-1}
  I have largely adopted the terminology and notation of
  \cite{MR2515846}, so I should point out a number of differences. In
  \cite[eq. 2]{MR2515846}, it is assumed that $Ψ$ is a complex-valued
  function (and ultimately real-valued for the minimally-coupled
  case), but this is not necessary for the mathematics. Similarly, the
  form of the potential
  $V(Ψ) = ½ m^2 \normof{Ψ}^2 + \frac{1}{24} λ \normof{Ψ}^4$ (with
  $λ=0$ in the minimally-coupled case) is
  used~\cite[eq. 3,12]{MR2515846}, but while this may make physical
  sense, it it not necessary for the mathematical results here.
  The density $ρ$ allows us to study integrability on an arbitrary
  energy level.
\end{remark}

\subsubsection{The Phase Space}
\label{sec:phase-space}

It is useful to clarify the phase space of the hamiltonians in
question. It makes mathematical sense to choose the largest space on
which the hamiltonians can be defined while simultaneously
preserving their algebraic character. On the other hand, the
physical origins of the model indicate that the locus $\set{a=0}$ is
one with special meaning and the model ceases to be meaningful near
this set. Belinsky, Khalatnikov \& Lifshitz met such concerns by
stating that general relativity is a purely gravitational theory and
their studies were meant to clarify that theory. Similar comments
are appropriate here. It is also important to note that the sign of
$a$ has no intrinsic meaning in the model and that the
\textem{correct} phase space is the quotient of $\set{(a,A,b,B)}$
obtained by identifying points $(a,A,b,B)$ and $(-a,-A,b,B)$. As is
so often the case, the behaviour of the hamiltonian $H$ in a
neighbourhood of the singular variety $\set{(0,0,b,B)}$ in the
reduced space contains a great deal of information and so we
de-singularize it to obtain that information. Or, in other words, we
simplify matters by studying $H$ on a de-singularized phase space
where the sign of $a$ is defined--but the conclusions must be
independent of this latter fact.

\section{Minimal coupling}
\label{sec:minimal-coupling}

Let us investigate the normal form for the minimally-coupled
hamiltonian $H$~\eqref{eq:min-couple}. Recall that by
assumption~\ref{it:a3}, $V=V_2+V_3+V_4$.

\begin{lemma}
  \label{lem:min-couple-resolve-singularity}
  The function $ν = ν(x,y,A,B) = xA + yB/x$ is a generating function
  of the symplectic transformation
  \begin{align}
    \label{eq:min-couple-gen-fun}
    a &= x, &&& b &= y/x, &&& A &= X, &&& B &= xY.
  \end{align}
  The hamiltonian \eqref{eq:min-couple} is transformed to
  \begin{equation}
    \label{eq:min-couple-H-resolved}
    H = -\underbrace{½ \left[ X^2 + k x^2 - ½ Λ x^4 \right]}_{H^{(1)}}
    + \underbrace{½ \left[ \normof{Y}^2 + V_4(y) + x V_3(y) + x^2 V_2(y) \right]}_{H^{(2)}}.
  \end{equation}
\end{lemma}
\begin{proof}
  By definition, the change of variables is defined from the equations
  \[ X = ν_x = A, \qquad Y = ν_y = B/x, \qquad a = ν_A = x, \qquad b = ν_B = y/x. \]
  This defines a symplectic change of variables; the remainder is clear.
\end{proof}

In the new coordinate system, courtesy of the assumption~\ref{it:a3}
on the potential $V$, the hamiltonian $H$ has forgotten the
singularity at $a=0(=x)$. Roughly speaking, the transformation has
glued $[0,∞) × \R^n$ and $(-∞,0] × \R^n$ along the singular variety
$\set{0} × \R^n$ to produce a copy of $\R × \R^n$ where the potential
of the system is regular.

\begin{lemma}
  \label{lem:min-couple-rescaling}
  Assume $k=±1$ (i.e. the scalar curvature of $g$ is $±6$) and
  $k Λ > 0$. Let $α^2=1/kΛ$, $ε>0$, and
  \begin{align}
    \label{eq:min-couple-rescaling}
    x &= α u, &&& X &= -α U, &&& y &= √ ε w, &&& Y &= √ ε W.
  \end{align}
  Then the hamiltonian differential equations of $H$ are transformed
to
\begin{align}
  \label{eq:min-couple-rescaled-des}
  u' &= U, &&& U' = u'' &= -ku \left( 1 - u^2 \right) + ε u V_2(w) + ½ ε^{\frac{3}{2}} α^{-1} V_3(w), \\\notag
  w' &= W, &&& W' = w'' &= -½ \left[ α^2 u^2 ∇ V_2(w) + α ε^{½} u ∇ {V}_3(w) + ε ∇ {V}_4(w) \right].
\end{align}
\end{lemma}

The proof of the lemma is a calculation. Note that without imposing
the equality $ε=α^2$, the DE are no longer in canonical form. That is
a price worth paying in order to examine the system near $y=Y=0$.

By hypothesis~\ref{it:a3}, the quadratic form $V_2$ is positive
definite. Therefore, there is an orthonormal change of variables such
that $V_2$ is transformed to a weighted sum of squares, weighted by
its eigenvalues. Since the linear change of variables does not affect
the structure of the DE~\eqref{eq:min-couple-rescaled-des}, it can be
assumed without loss of generality that
\begin{hypotheses}
\item\label{it:a6} the quadratic form $V_2$ equals
  \begin{align}
    \label{eq:min-couple-v2}
    V_2(w) &= ½ k Λ \ip{φ(w)}{φ(w)}, &&& \textrm{where\ }φ &= \diagmatrix{φ_1}{φ_n}>0.
  \end{align}
\end{hypotheses}

To make this a regular perturbation problem, one can assume either
\begin{althypotheses}
\item\label{it:a7} $V_3=O(√ ε)$, i.e. $V_3=√ ε \tilde{V}_3$ for some
  homogeneous cubic $\tilde{V}_3$; or
\item\label{it:a8} $V_3 ≡ 0$, i.e. the potential function is even.
\end{althypotheses}
And, in all cases, $k Λ > 0$ is fixed.

\begin{lemma}
  \label{lem:min-couple-decoupled-eps=0}
  Assume \ref{it:a6} and either \ref{it:a7} or \ref{it:a8}. Then, the
  system of DE~\eqref{eq:min-couple-rescaled-des} is quadratic in the
  parameter $ε$ and for $ε=0$, the system is transformed to
  \begin{align}
    \label{eq:min-couple-decoupled-eps=0}
    u' &= U, &&& U' = u'' &= -ku(1-u^2), \\\notag
    w' &= W, &&& W' = w'' &= - ½ φ^2 \left[ 1 - \left( 1 - u^2 \right)\right] w.
  \end{align}
\end{lemma}
At this point, the treatment of the cases $k=1$ and $-1$ diverge
somewhat. The case of $k=1$ is treated first.

\subsection{$k=1$}
\label{sec:min-couple-k=1}

The system~\eqref{eq:min-couple-decoupled-eps=0} has a pair of
saddle-centre critical points at $(u=± 1, U=0, w=W=0)$. Moreover,
for $σ=± 1$ the hyper-planes
\begin{align}
  \label{eq:min-couple-nhim}
  N_{σ} &= \set{(u=σ,U=0,w,W) \mid w,W ∈ \R^n }
\end{align}
are normally hyperbolic invariant manifolds that are foliated by
invariant tori. Let $\unstablemanifold[+/-]{N_{σ}}$ be the
stable/unstable manifolds of $N_{σ}$. A connected component of the
stable manifold of $N_{σ}$ less $N_{σ}$,
i.e. $\stablemanifold{N_{σ}}-N_{σ}$, coincides with a connected
component of the unstable manifold of $N_{-σ}$ less $N_{-σ}$,
$\unstablemanifold{N_{-σ}}-N_{-σ}$ (see
figure~\ref{fig:min-couple-separatrix}). These invariant manifolds are
contained in the zero set of the function $H^{(1)} = H_1^{(1)}$ where
\begin{equation}
  \label{eq:min-couple-h1}
  H_k^{(1)}(u,U,w,W) = ½ U^2 - \dfrac{k}{4} (1-u^2)^2,
\end{equation}
which is an integral of motion
of~\eqref{eq:min-couple-decoupled-eps=0}.

For $ε>0$ sufficiently small, the local normally hyperbolic invariant
manifolds $N_{σ}^{loc}$ will be perturbed to normally hyperbolic
invariant manifolds $N_{σ,ε}^{loc}$ that are invariant for the flow of
\eqref{eq:min-couple-rescaled-des} and are graphs over
$N_{σ}^{loc}$. The stable and unstable manifolds,
$\stablemanifold{N_{σ,ε}^{loc}}$ and
$\unstablemanifold{N_{-σ,ε}^{loc}}$, will generally no longer ``coincide''
as for $ε=0$. Since they are codimension-1 submanifolds, the distance
between them can be measured by $H^{(1)}$. Specifically, the
\poincaremelnikov{} function measures the $O(ε)$ separation between
the two perturbed invariant manifolds. In the present case,
\begin{equation}
  \label{eq:min-couple-melnikov-poincare}
  M(P) = \int\limits_{-∞}^{∞} U(t)\, u(t)\, V_2(w(t))\, \D{t} = ½ \sum_{j=1}^n φ_j^2 \, \int\limits_{-∞}^{∞} w_j(t)^2 \, \ddt{}{t} \left( u^2 - σ^2 \right) \D{t},
\end{equation}
where $P=(u(0),U(0),w(0),W(0)) ∈ \stablemanifold{N_{σ}} \cap
\unstablemanifold{N_{-σ}}$ and $P(t)=(u(t),U(t),w(t),W(t))$ is the
solution to the unperturbed DE~\eqref{eq:min-couple-decoupled-eps=0}.

\subsection{$k=-1$}
\label{sec:min-couple-k=-1}

In this case, the origin $(u=0,U=0,w=0,W=0)$ is a saddle-degenerate
centre equilibrium of the DE~\eqref{eq:min-couple-rescaled-des}. In
this case, the remainder of the discussion in the previous subsection
carries over with $σ=0$ in place of $σ=± 1$ and $k=-1$ in place of
$k=1$. In particular, the \poincaremelnikov{}
function~\eqref{eq:min-couple-melnikov-poincare} describes the $O(ε)$
separation of the local stable and unstable manifolds of the local
normally hyperbolic invariant manifold $N^{loc}_{0,ε}$.

\medskip

The computation of the \poincaremelnikov{}
integral~\eqref{eq:min-couple-melnikov-poincare} is deferred to
section~\ref{sec:computation}.

\section{Minimal conformal coupling}
\label{sec:min-conf-couple}

\begin{lemma}
  \label{lem:min-conf-couple}
  Assume \ref{it:a3}. Then the hamiltonian of the minimal conformal
  coupling model \eqref{eq:con-couple} is transformed to the sum of
  the hamiltonian of the minimal coupling model
  \eqref{eq:min-couple-H-resolved} and $k × ½ \normof{y}^2$ under the
  identity transformation $x=a, X=A, y=b, Y=B$.
\end{lemma}

\subsection{$k=1$}
\label{sec:min-conf-couple-k=1}

The lemma implies that for $k=1$, the minimal coupling and minimal
conformal coupling models are virtually identical--the sole change
being that the equation for $w''$
in~\eqref{eq:min-couple-decoupled-eps=0} has an additional term of $-w$
on the right-hand side.

\subsection{$k=-1$}
\label{sec:min-conf-couple-k=-1}

However, for $k=-1$, the treatment of the models diverges somewhat
because the origin in the latter model is a non-degenerate saddle
equilibrium. In this case, we are forced to make some additional
assumptions about the potential $V$. A minimal requirement is that
$V_4$ be positive-definite so that $H^{(2)}$ is proper and $V_3 ≡ 0$
so that $\set{x=X=0}$ is invariant. On the other hand, if $H^{(2)}$ is
non-integrable on the plane $\set{x=X=0}$, then there is nothing to be
proven, so we make the assumption that $H^{(2)}$ is integrable on this
plane, too. Finally, as the discussion of the integrable cases of the
generalized Yang-Mills hamiltonian in the introduction makes clear,
there are very few known integrable cases. The only two cases that are
relevant to the particular problem here are the cases where $V$ is
either rotationally invariant (under the action of $\SOrth{n}$) or
separable; likewise $H$ should be separable.

\begin{hypotheses}
\item\label{it:a13} the potential $V_{ε} = ε V_2 + ε^2 V_3 + V_4$,
  $V_4$ is positive and either rotationally-invariant or separable while
  $V_2$ satisfies \ref{it:a6}.
\end{hypotheses}

If hypothesis~\ref{it:a13} is assumed, then for $ε=0$, the hamiltonian
$H_{ε}^{(2)}$ has a saddle critical point at $y=Y=0$. It follows that
the hamiltonian $H_{ε} = -H^{(1)} + H_{ε}^{(2)}$ has a saddle critical
point $s$ at $x=X=0$, $y=Y=0$ that persists for all $ε$ (to be clear,
the saddle point for $H_{ε}$ is denoted by $s_{ε}$ below). By the
hypothesis~\ref{it:a13}, the stable and unstable manifolds
$\stablemanifold[±]{s_{ε}}$ are coincident, lagrangian submanifolds
for $ε=0$. For $ε$ non-zero and sufficiently small, the local
manifolds are lagrangian graphs over the unperturbed local
manifolds. Since the local stable manifold is contractible, the theory
of lagrangian submanifolds implies that there are analytic functions
$ν^{±}_{ε} : \localstablemanifold{s} → \R$ such that
$\localstablemanifold[±]{s_{ε}}$ is the graph of $ν^{±}_{ε}$ and
$ν_{ε} = ν^+_{ε} - ν^-_{ε} = ε ν_0 + O(ε^2)$.

\begin{definition}
  \label{def:pmp}
  The function $ν_{ε}$ described in the previous paragraph is called
  the \poincaremelnikov{} splitting potential; the function $ν_0$ is
  its lowest-order term.
\end{definition}

A critical point of $ν_{ε}$ is a point of intersection of
$\localstablemanifold{s_{ε}}$ and $\localunstablemanifold{s_{ε}}$;
since the hamiltonian vector field $X_H$ is tangent to each manifold,
such a critical point is not isolated but instead lies on a smooth
curve of critical points. The maximal rank of the hessian
$\hess{ν_{ε}}$ at such a critical point is therefore $n$ and at such
points, the local manifolds intersect transversely as submanifolds of
the common energy level (i.e. they are each $n+1$ dimensional
submanifolds of a $2n+1$ dimensional iso-energy manifold that
intersect along a curve).

\begin{proposition}
  \label{prop:pmp}
  Assume that $ν_0$ has a critical point at
  $P ∈ \localstablemanifold{s}$. If the hessian $\hess{ν_0}$ has rank
  $n$ at $P$, then the perturbed stable and unstable manifolds
  intersect transversely at a nearby point $P_{ε}=P+O(ε)$.
\end{proposition}

As mentioned, transversality means the submanifolds intersect
transversely in the energy level. The proof of this proposition may be
reconstructed along the lines of \cite[Theorem 3.4]{MR2554208}; see
also~\cite{MR1766491} for an exposition. The lowest-order term $ν_0$
can be computed from the decomposition of $H_{ε}=H_0+ε H_1$ by
\begin{equation}
  \label{eq:pmp}
  ν_0(P) = \int\limits_{-∞}^{∞} H_1 \circ \phi^t(P)\, \D{t},
\end{equation}
where $\phi^t$ is the hamiltonian flow of $H_0$ and $P ∈
\localstablemanifold[±]{s}$. 

\begin{proposition}
  \label{prop:pmp-k=-1}
  Assume hypothesis \ref{it:a13}. If
  \begin{enumerate}
  \item $V_4$ is rotationally invariant; or
  \item $V_4$ is separable and the polynomial $Δ$~\eqref{eq:pmp-d2nu-non-deg-(2)} is non-zero,
  \end{enumerate}
  then $\localstablemanifold{s_{ε}}$ and $\localunstablemanifold{s_{ε}}$ do not coincide for all $ε ≠ 0$ sufficiently small.

  If
  \begin{enumerate}
  \item $V_4$ is rotationally invariant and the eigenvalues of $φ$ are distinct; or
  \item $V_4$ is separable and the polynomial $Δ$~\eqref{eq:pmp-d2nu-non-deg-(2)} does not vanish at $5$,
  \end{enumerate}
  then $\localstablemanifold{s_{ε}}$ and $\localunstablemanifold{s_{ε}}$ intersect transversely for all $ε ≠ 0$ sufficiently small.
\end{proposition}

\begin{proof}
  In both cases, the lowest-order term $ν_0$ of the \poincaremelnikov{} potential is
  \begin{equation}
    \label{eq:pmp-V2}
    ν_0(P) = ½ L × \int\limits_{-∞}^{∞} x(t)^2 \, \normof{φ(y(t))}^2 \, \D{t}
  \end{equation}
  where $P=(x(0),X(0),y(0),Y(0)) ∈ \localstablemanifold[±]{s_0}$ and
  $P(t)=(x(t),X(t),y(t),Y(t))$ is the solution to the hamiltonian DEs
  for $H_0$ and $L = k Λ > 0$.

  \medskip
  
  (1) Assume that $V_4$ is rotationally invariant, so
  $V_4(y) = ¼ λ \normof{y}^2$ for some $λ > 0$. Let us change
  variables
  \begin{align}
    \label{eq:pmp-change-of-vars-(1)}
    x &= α u, &&& y &= β r θ, &&\textrm{ where }& α^{-2} &= L, &&& β^{-2} &= λ,
  \end{align}
  and $r ∈ [0,∞)$ and $θ ∈ \sphere{n-1}$. Because $H^{(2)}$ is
  rotationally invariant, the momentum map
  $Ψ = ½ \left( y Y^T - Y y^T \right)$ is a first integral which
  vanishes identically on $\localstablemanifold[±]{s_0}$. This implies
  that the hamiltonian DEs of $H_0$ are transformed, along the saddle,
  to
  \begin{align}
    \label{eq:pmp-ham-des-(1)}
    u'' &= u(1-u^2), &&& r'' &= r(1-r^2), &&& θ' &= 0.
  \end{align}
  Let $μ(t)$ be an even solution to the DE for $u$; then the solutions
  $u=± μ(t-t_0)$ and $r=± μ(t-t_1)$ for some $t_0, t_1$. Finally, with
  $τ = t_1-t_0$,
  \begin{equation}
    \label{eq:pmp-nu0-(1)}
    ν_0(P) = ν_0(τ,θ) = ½ β^2 \normof{φ(θ)}^2 \, \int\limits_{-∞}^{∞} μ(s)^2 \, μ(s-τ)^2 \, \D{s}.
  \end{equation}
  The choice of the solution $μ$ has induced coordinates $(t_0,τ,θ)$
  on $\localstablemanifold[±]{s_0}$. The invariance of $ν_0$ under the
  unperturbed flow means it depends only on $(τ,θ)$.

  By even/odd symmetry, $\didi{ν_0}{τ} = 0$ at $τ=0$. Thus, if $θ_0$
  is a maximum point of $θ → \normof{φ(θ)}^2$, then $(τ=0,θ=θ_1)$ is a
  critical point of $ν_0$ (and a simple argument using the
  Cauchy-Schwarz inequality shows it is a global maximum). On the
  other hand,
  \begin{equation}
    \label{eq:pmp-d2nu0-(1)}
    \ditdit{ν_0}{τ} \Biggr|_{τ=0} = ½ β^2 \normof{φ(θ)}^2 \, \int\limits_{-∞}^{∞} μ(s)^4 \, \left( (2-μ(s)^2)/2 + (1-μ(s)^2 \right) \, \D{s} = - \dfrac{32}{15} β^2 \normof{φ(θ)}^2.
  \end{equation}
  This proves case (1), since the hessian of $ν_0$ is non-trivial at
  $(τ=0,θ=θ_1)$ (hence $ν_0 ≠ $ constant) and it has rank $n$ if the
  eigenvalues of $φ$ are distinct.

  \medskip
  
  (2) Assume that $V_4$ is separable, so
  $V_4(y) = ¼ \sum_{i=1}^n λ_i y_i^4$ for some positive scalars
  $λ_1, \ldots, λ_n$. Let us note that the separating coordinates are
  not necessarily the coordinates in which $φ$ is diagonal. Let us
  change variables
  \begin{align}
    \label{eq:pmp-change-of-vars-(2)}
    x &= α u, &&& y &= β w, &&\textrm{ where }& α^{-2} &= L, &&& β^{-2} &= \diagmatrix{λ_1}{λ_n}.
  \end{align}
  The change of coordinates implies that the hamiltonian DEs of $H_0$
  are transformed to
  \begin{align}
    \label{eq:pmp-ham-des-(2)}
    u'' &= u(1-u^2), &&& w_i'' &= w_i(1-w_i^2), &&& i&=1,\ldots,n.
  \end{align}
  Let $Φ = (φ β)^2$ and $μ$ be as above. Then, similar to case (1),
  \begin{equation}
    \label{eq:pmp-nu0-(2)}
    ν_0(P) = ν_0(τ) = ½ \sum_{i,j=1}^n Φ_{ij} \ \int\limits_{-∞}^{∞} μ(s)^2 \, μ(s-τ_i) \, μ(s-τ_j) \, \D{s}.
  \end{equation}
  Similar to case (1), $τ=0$ is always a critical point. Let $σ(Φ)$ be
  the diagonal matrix whose $i,i$ entry is the sum of the elements in
  the $i$-th row of $Φ$; define $δ(Φ,z)=-Φ+z σ(Φ)$. Calculations
  similar to case (1) show that
  \begin{equation}
    \label{eq:pmp-d2nu0-(2)}
    \hess{ν_0}\Bigr|_{τ=0} = - \dfrac{8}{5} × δ(Φ,5).
  \end{equation}
  Hence, the critical point $τ=0$ is non-degenerate iff
  \begin{equation}
    \label{eq:pmp-d2nu-non-deg-(2)}
    Δ(z) = \det{δ(Φ,z)}
  \end{equation}
  is non-zero at $z=5$. To investigate when $\hess{ν_0}$ is
  non-trivial, assume $Δ(z)$ is a non-zero constant multiple of
  $(z-5)^n$. Then, since $Φ$ is positive definite, the characteristic
  polynomial of $Φ$ is $(z-1/5)^n$. Since $Φ$ is symmetric, this
  forces $Φ=5 σ(Φ)$, so $Φ$ is diagonal and therefore $Φ$ is zero. But
  $Φ$ is positive definite, a contradiction. This proves case (2).
\end{proof}



\begin{remark}
  \label{rem:pmp-d2nu}
  In the decoupled ($ε=0$) limit the plane $\set{y=Y=0}$ is a
  normally-hyperbolic invariant manifold and, aside from the saddle at
  the origin and its connections, the plane is fibred by periodic
  orbits (similarly, $\set{x=X=0}$ is foliated by normally-hyperbolic
  invariant tori). A consequence of Proposition~\ref{prop:pmp-k=-1} is
  that the homoclinic connections split for energy levels close to $0$
  and all $ε ≠ 0$ sufficiently small. Hence, the hamiltonian flow
  enjoys horseshoes on all energy levels near the zero level.

  One can be more quantitative: the lowest-order term $ν_0$ in the
  \poincaremelnikov{} splitting potential is an explicit integral
  involving the function $μ$ and the Jacobi dn (negative energy) or cn
  function. In case (1), the expression is the same
  as~\eqref{eq:pmp-nu0-(1)} except the
  integral is changed to
  \begin{equation}
    \label{eq:pmp-d2nu-(1)-alt}
    ν_0(τ) = \int\limits_{-∞}^{∞} u(s)^2 μ(s-τ)^2 \, \D{s}
  \end{equation}
  where $u$ is a solution to the first DE
  in~\eqref{eq:pmp-ham-des-(1)}. Figure~\ref{fig:pmp-d2nu-(1)}
  plots the second derivative of $ν_0(τ)$ at the critical point $τ=0$.
\end{remark}

\beginefigure
  \caption{The second derivative of $ν_0$ with respect to $τ$ at $τ=0$
    as a function of
    $H^{(1)} := H_{-1}^{(1)}$~\eqref{eq:min-couple-h1}. The hessian
    vanishes at the saddle-centre where $H^{(1)}=-1/4$; the indicated
    point at zero energy coincides with the value
    in~\eqref{eq:pmp-d2nu0-(1)}. Inset: the contours of $H^{(1)}$ in
    the $(u,u')$ plane.}
  \label{fig:pmp-d2nu-(1)}
  \includegraphics[width=\efigurewidth,keepaspectratio]{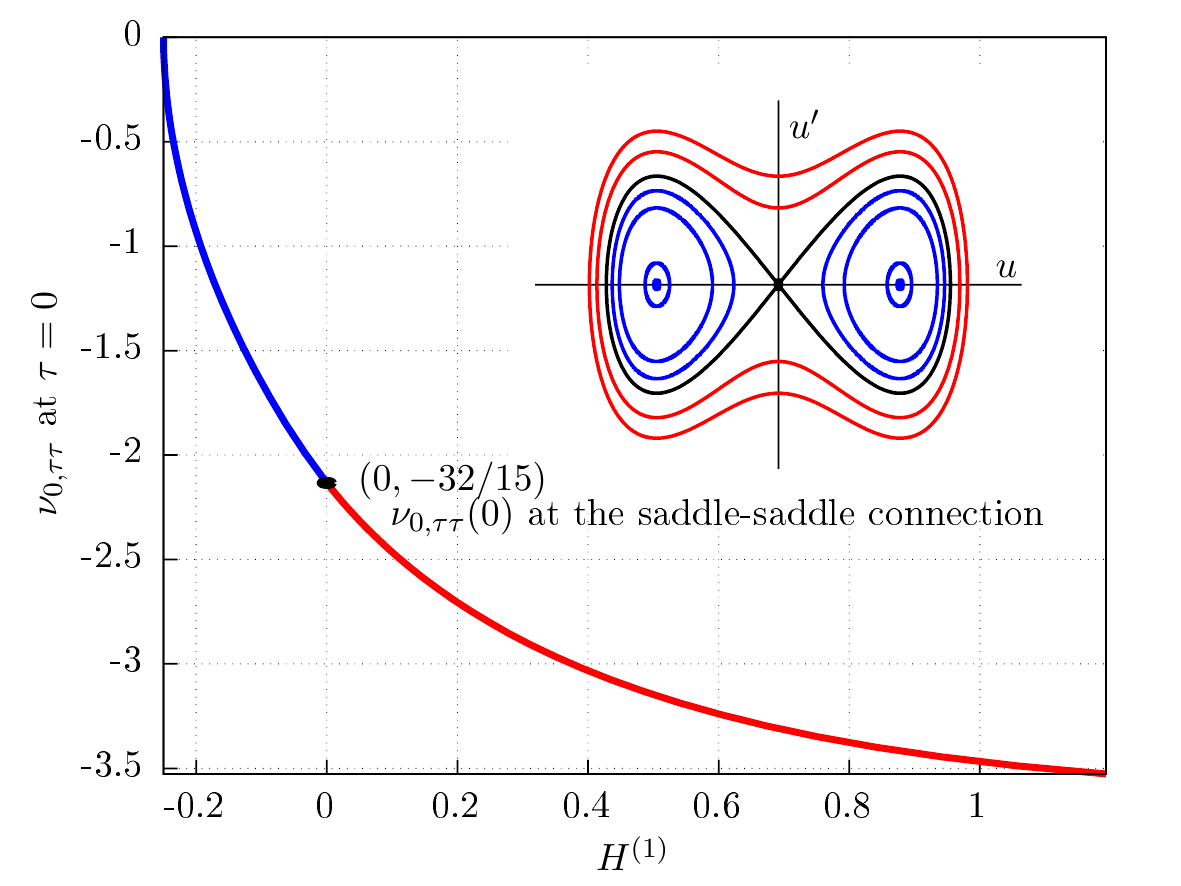} 
\endefigure

\begin{proof}[Proof of Theorem~\ref{thm:proof-of-conf-couple-k=-1}]
  Let $V_4(b)=½ λ \normof{b}^4$ and $V_2(b)=ε (m a \normof{b})^2$ for
  $b ∈ \R^2$ where $λ<0$ and $m > 0$ are fixed. By
  Proposition~\ref{prop:pmp-k=-1} the stable and unstable manifolds of
  the saddle fixed point at the origin intersect transversely modulo
  rotations for all $ε ≠ 0$ sufficiently small. Therefore, the nearby
  hyperbolic periodic tori in the $\set{a=A=0}$ plane have stable and
  unstable manifolds that intersect transversely modulo rotations,
  too. If the angular momentum $ω$ is fixed and small enough, then the
  reduction of the hamiltonian yields a hamiltonian in the form of
  \eqref{eq:maciej-conf-couple-H} and the normally hyperbolic tori are
  reduced to hyperbolic periodic orbits with transverse homoclinic
  points. This implies the theorem.
\end{proof}

\section{Computation of the \poincaremelnikov{} integral}
\label{sec:computation}

Let us explain how the \poincaremelnikov{}
integral~\eqref{eq:min-couple-melnikov-poincare} is computed. The
integral is of the form
\begin{align}
  \label{eq:comp-form-of-des}
  M(P)                                   & = \int\limits_{-∞}^{∞} \dash{q}(t) w(t)^2\, \D{t} &  &  & \textrm{subject to } \\
  \label{eq:comp-form-w-de}
  \ddash{w} + \left[ β^2 - q(t) \right] w & = 0,
\end{align}
where $q(t)$ and $\dash{q}(t)$ both vanish at $t=± ∞$. To compute the
integral $M$, let us make the simplifying assumption that $q$ is an
even function. If $w_0$ (resp. $w_1$) is the unique solution such that
$w_0(0)=1, \dash{w}_0(0)=0$ (resp. $w_1(0)=0, \dash{w}_1(0)=1$), then
the general solution $w = c_0 w_0 + c_1 w_1$ and
$M(P) = 2 c_0 c_1 m_{01}$ where
$m_{01} = \int_{-∞}^{∞} \dash{q}(t) w_0(t) w_1(t)\, \D{t}$. That is, in
the coordinates $(t_0,c_0,c_1)$, $M$ is either identically zero
($m_{01}=0$) or it is an indefinite quadratic form and therefore its
zero locus is $\set{c_0c_1=0}$ and $\D{M} ≠ 0$ on the zero locus
except at $c_0=c_1=0$.

The computation diverges somewhat depending on the value of $k$.

\subsection{$k=-1$}
\label{sec:computation-k=-1}

In this case, the minimal coupling model's
DE~\eqref{eq:min-couple-decoupled-eps=0} and the \poincaremelnikov{}
integral~\eqref{eq:min-couple-melnikov-poincare} translate to
$q=-½ φ^2 u^2$ and $β=0$ in
equations~(\ref{eq:comp-form-of-des},\ref{eq:comp-form-w-de}). Since
$β=0$, the fundamental solutions are $w_0=1$ and $w_1=t$. Then,
integration by parts gives
$m_{01} = - ½ φ^2 \int_{-∞}^{∞} u^2 \D{t} =  - 2 φ^2$.

\subsection{$k=1$}
\label{sec:computation-k=1}

To compute $M$ explicitly in this case it is more convenient to
complexify. Let $w^{τ}_{σ}(t)$ be solutions to the
DE~\eqref{eq:comp-form-w-de} such that $w^{τ}_{σ}(t)$ is asymptotic to
$\exp((-1)^σ i β t)$ at $t=(-1)^τ ∞$ for $τ,σ ∈ \set{0,1}$. Then,
there is a unique change of basis $a=[a_{ij}]$ from
$\set{w^1_0,w^1_1}$ to $\set{w^0_0,w^0_1}$ such that
$w^0_i = a_{i0} w^1_0 + a_{i1} w^1_1$ for $i=0,1$. Let
$w = c_0 w^1_0 + c_1 w^1_1$ be an expansion of the solution $w$ in
terms of the basis $\set{w^1_0,w^1_1}$ of solutions. It is proven
in~\cite[Corollary 3.1]{MR2425326} that the integral $M$ equals
$m_{00} c_0^2 + 2 m_{01} c_0 c_1 + m_{11} c_1^2$ where the constants
$m_{ij}$ are the coefficients of the complexified \poincaremelnikov{}
\textem{form} $M$ with respect to this basis. These coefficients are
explicitly calculable in terms of the scattering matrix $a$. Due to
the fact that the real form of $M$ is either zero or indefinite, it
suffices to prove that the complexified $M$ has a non-zero determinant
(it will necessarily be positive).

To compute the determinant of $M$ in the present situation, one
rewrites the DE~\eqref{eq:min-couple-decoupled-eps=0} with the
variable $u$ as an independent variable in place of $t$ (since along
the separatrix solution $u$ is monotone). In that case the DE for $w$
is transformed to a Legendre DE:
\begin{equation}
  \label{eq:comp-legendre}
  (1-u^2) \ddt[2]{w}{u} - 2u \ddt{w}{u} + \left( ν(ν+1) - \dfrac{μ^2}{1-u^2} \right)w = 0.
\end{equation}
One sees that
\begin{enumerate}
\item\label{it:comp-min-couple} in the minimal coupling case (with
  $β=φ$), $ν=\frac{-1 ± \sqrt{1-4 φ^2}}{2}$ and $μ = ± i φ$;
\item\label{it:comp-conf-couple} in the minimal conformal coupling
  case (with $β=\sqrt{2+φ^2}$), $ν=\frac{-1 ± \sqrt{1-4 φ^2}}{2}$ and
  $μ = ± i \sqrt{2+φ^2} = ± i β$.
\end{enumerate}
It can be demonstrated that the connection matrix for
\eqref{eq:comp-legendre} is
\begin{align}
  \label{eq:comp-connection-matrix}
  [a_{ij}]     & =
  \begin{bmatrix}
    \bar B     & 2^{i β}\bar A\                                                                                                                                                  \\
    2^{-i β} A & B
  \end{bmatrix}
               &  &  & \textrm{where } &  &  &  & A & = \dfrac{Γ({\rmc}) Γ(1-{\rmc})}{Γ({\rma}) Γ({\rmb})}, B = \dfrac{Γ({\rmc}) Γ({\rmc}-1)}{Γ({\rmc}-{\rma}) Γ({\rmc}-{\rmb})} \\\notag
               &  &  &                 &  &  &  &   & {\rma} & = -ν, {\rmb}=1+ν, {\rmc}=1-μ.
\end{align}
From this, the determinant of the complexified \poincaremelnikov{} quadratic form $M$
equals
\begin{equation}
  \label{eq:comp-detI}
  \det M = -4 β^4 \left( 1 - (|A|+|B|)^2 \right)\left( 1 - (|A|-|B|)^2 \right) = 16\, β^4 \normof{A}^2,
\end{equation}
where the fact that the connection matrix has unit determinant is used.

\begin{proposition}
  \label{prop:melnikov-does-not-vanish}
  The following holds when $k=1$:
  \begin{enumerate}
  \item\label{it:mel-min-couple} in the minimal coupling model, for
    all $φ=β>0$, $\det M ≠ 0$;
  \item\label{it:mel-conf-couple} in the minimal conformal coupling
    model, for all $φ>0$ ($β>√ 2$), $\det M ≠ 0$.
  \end{enumerate}
\end{proposition}

\begin{proof}
  \label{pf:melnikov-does-not-vanish}
  Assume that $ν = - ½ + \sqrt{¼ - φ^2}$ and $μ = i β$ are independent parameters.
  
  First, assume that $φ>½$. In this case, the proof is similar to that
  of \cite[Theorem 4.1]{MR3265722}: let ${\rma}=½ + i s$ where
  $s=\sqrt{φ ² - ¼}$ is real and positive. Equations
  (\ref{eq:comp-connection-matrix}--~\ref{eq:comp-detI}) and
  \cite[6.1.28--30]{MR1225604} imply that for $s, β > 0$
  \begin{equation}
    \label{eq:connection-coeffs-in-s-beta-1}
    \normof{A}^2 = \left( \dfrac{\cosh π s}{\sinh π β} \right)^2.
  \end{equation}

  Second, assume $0<φ<½$. Then ${\rma} = ½ + s$ and ${\rmb} = ½ - s$ where
  $s=\sqrt{¼ - φ ²} ∈ (0,½)$ and so ${\rma}, {\rmb} ∈ (0,1)$. In this case, the
  reflection formulae cited above along with
  $Γ(z) Γ(1-z) = π / \sin(π z)$ for
  $z ∈ (0,1) + i \R$~\cite[6.1.17]{MR1225604} imply that for $β>0$ and
  $0<s<½$
  \begin{equation}
    \label{eq:connection-coeffs-in-s-beta-2}
    \normof{A}^2 = \left( \dfrac{\cos π s}{\sinh π β} \right)^2,
  \end{equation}
  which does not vanish for $s ∈ (0,½)$.

  Finally, it is apparent that $\det M ≠ 0$ when $s=0$ ($φ=½$), too.

  This proves case~\ref{it:mel-min-couple} where
  $β=φ>0$. Case~\ref{it:mel-conf-couple} follows since $φ>0$ and $β=\sqrt{2+φ^2}$.
\end{proof}

\begin{remark}
  \label{rem:melnikov-does-not-vanish}
  Figures~\ref{fig:detm-min} and \ref{fig:detm-conf} graph the
  re-scaled determinant of the complex \poincaremelnikov{} form versus
  $β$ for the minimal and minimal conformal coupling models. Note that
  although $β > √ 2$ in the latter, the graph is extended over the
  interval $[0,√ 2]$ where one sees that the integrable case of $m^2/Λ
  = -1$ is identified at $β=0$.

  The literature on saddle-centre equilibria in hamiltonian systems is
  extensive. Lerman~\cite{MR987443} and Lerman \& Kol$'$tsova
  \cite{MR1316697,MR1341274,MR1409407} study the general case of a
  2-degree of freedom hamiltonian with a saddle-centre equilibrium,
  and subsequently an $n+1$-degree of freedom hamiltonian with an
  equilibrium that decomposes as a saddle and $n$ centres. In the
  former case, they prove that for a generic hamiltonian, there is
  family of nearby hyperbolic periodic orbits which enjoy a pair of
  transverse homoclinic orbits; and in the latter case, similar
  results hold. In the $n+1$-degree of freedom setting, the generic
  case is there is a Cantor family of normally hyperbolic invariant
  tori near the saddle-centre and these tori enjoy transverse
  homoclinic orbits~\cite[Theorem
  1]{elsevier_sdoi_10_1016_0022_0396_74_90086_2}. Moreover, under
  generic conditions, one can demonstrate the existence of transition
  chains of tori and Arnol'd
  diffusion~\cite{iop10.1088/0951-7715/18/3/020,MR0163026}. Grotta-Ragazzo~\cite{MR1309550}
  gives an alternative, geometric, proof of the result of Lerman \&
  Koltsova and derives several corollaries from that proof. The first
  corollary is that, in the notation here, the separatrixes of the
  homoclinic connection split if the connection matrix $[a_{ij}]$ is
  not diagonal (i.e. $A ≠ 0$)~\cite[Theorem 4]{MR1309550} which is
  exactly the condition here that the complexified \poincaremelnikov{}
  form $M$ be non-degenerate~\eqref{eq:comp-detI}. The results of
  \cite{MR1309550} extend well beyond this, though. In \S4, the paper
  connects the non-triviality of $A$ with the non-triviality of
  squares in the monodromy group of the variational
  equation~\cite[Theorem 8]{MR1309550}. This point of view is
  elaborated in subsequent papers by Morales-Ruiz \& Peris and
  Yagasaki~\cite{MR1721562,iop10.1088/0951-7715/16/6/307} where the
  differential galois group is brought in. In recent work, Giles, Lamb
  \& Turaev~\cite{MR3551059} revisit the saddle-centre problem and
  derive a novel proof of the splitting result using Lyapunov-Schmidt
  reduction and the \poincaremelnikov{} potential. On the other hand,
  the classic work of Holmes \& Marsden, although couched in slightly
  different terminology, proves the existence of horseshoes on all
  super-energy levels in a neighbourhood of a saddle-centre separatrix
  under the assumption that the unperturbed hamiltonian is
  separable~\cite[Theorem 3.2 \& Example 4.1]{MR641913}.
\end{remark}

\beginefigure
  \caption{The contours of $H^{(1)}_1$~\eqref{eq:min-couple-h1} with the
    saddle fixed points and the separatrixes in blue. The arrows
    indicate the direction of the flow lines.}
  \label{fig:min-couple-separatrix}
  \includegraphics[width=\efigurewidth,keepaspectratio]{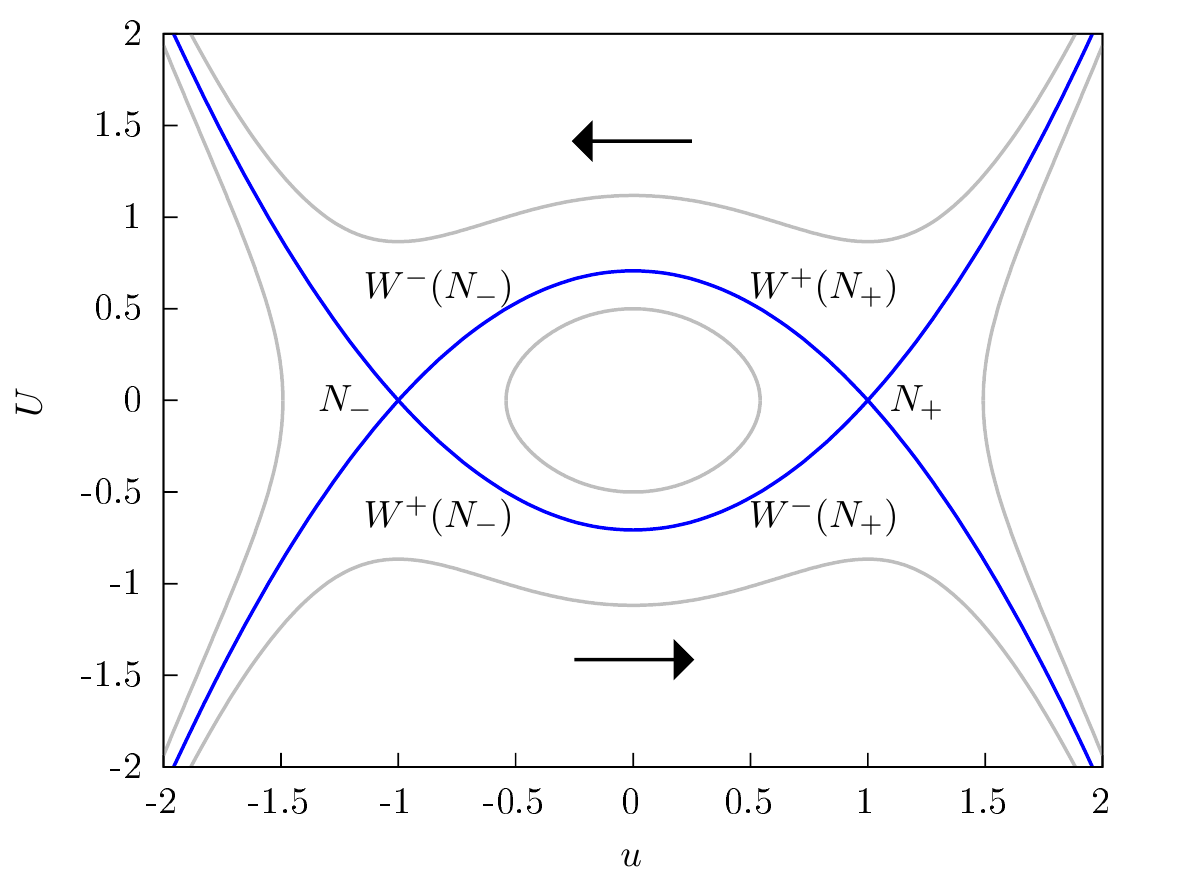}
\endefigure

\beginefigure
\caption{Minimal coupling: The graph of the scaled determinant of the
  Melnikov form vs. $β$. The inset shows the behaviour between $β=0$
  and $β=√ 2$.}
  \label{fig:detm-min}
  \includegraphics[width=\efigurewidth,keepaspectratio]{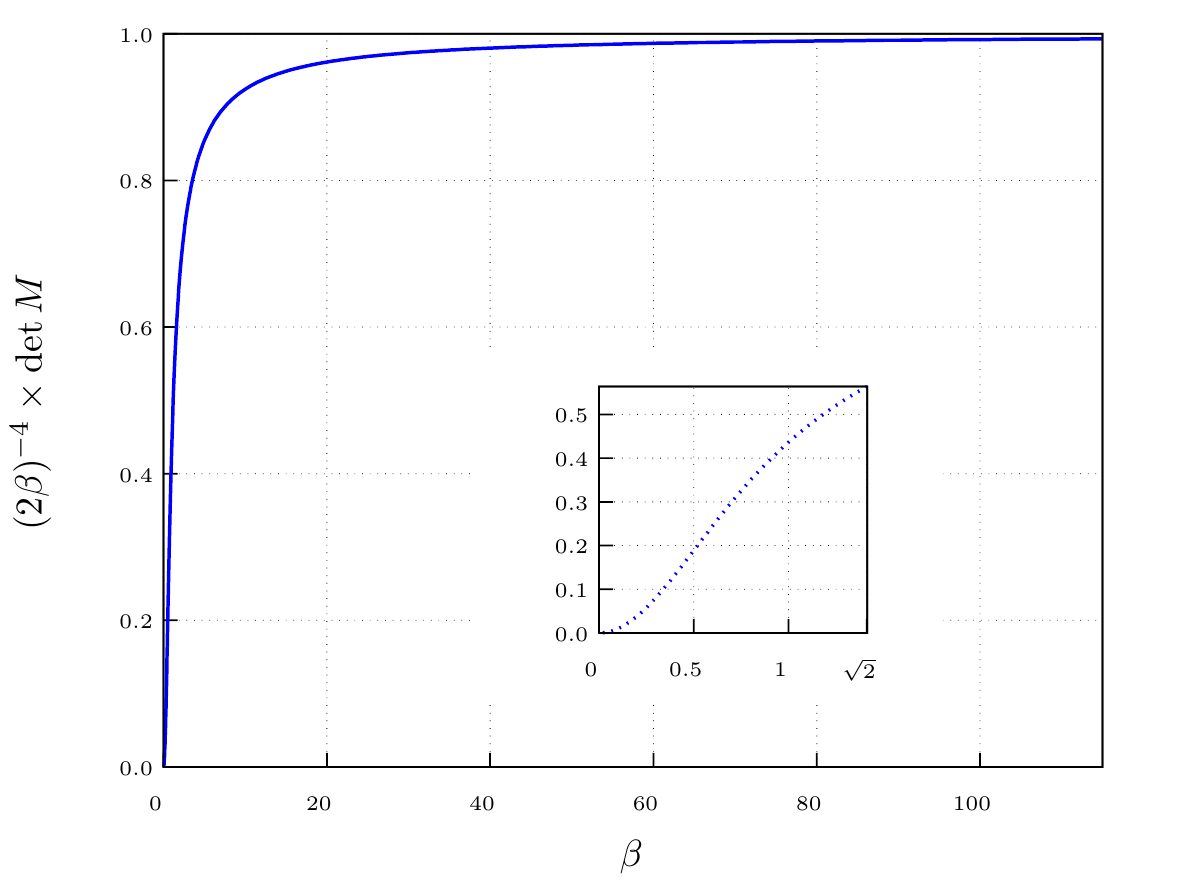}
\endefigure

\beginefigure
\caption{Minimal conformal coupling: The graph of the determinant of
  the Melnikov form vs. $β$. The inset shows the behaviour between
  $β=0$ ($m^2=-Λ$) and $β=√ 2$ ($m^2=0$).}
  \label{fig:detm-conf}
  \includegraphics[width=\efigurewidth,keepaspectratio]{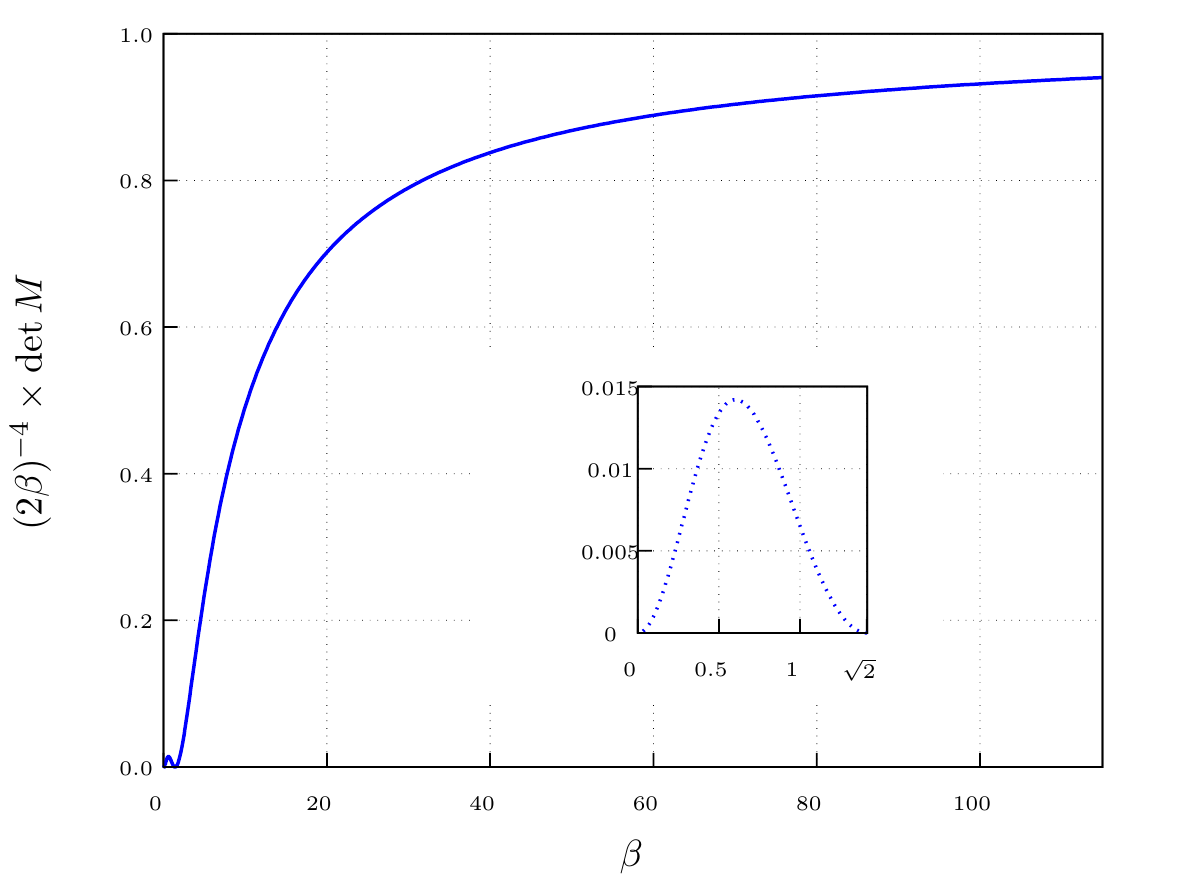}
\endefigure

\section{KAM tori}
\label{sec:kam-tori}

This section investigates the near-integrability of certain
hamiltonians that originate from the minimal coupling model. The
results here are intended to be illustrative rather than
exhaustive. The paper by Palacián,
et. al.~\cite{iop10.1088/1361-6544/ab1bc6} contains more expansive
results in a more specialized context.

Let us revisit the minimal coupling
hamiltonian~\eqref{eq:min-couple-H-resolved} with $k=-1$ and
$L=k Λ > 0$. In addition, the following hypothesis is assumed
\begin{hypotheses}
  \item\label{it:a14} the field is scalar, i.e. $n=1$.
\end{hypotheses}
In this case, the homogeneous terms in the potential of the scalar
field can be written as $V_i(y) = v_i y^i$ for $i=1,2,3$. By
hypothesis~\ref{it:a6}, $v_2 = ½ L φ ²$.

With these hypotheses, it follows that there is an elliptic critical
point of $H$ at $(x=α=1/√ L, X=0, y=0, Y=0)$.

\begin{lemma}
  \label{lem:kam-min-couple-rescaling}
  Assume $k=-1$ and $L = k Λ > 0$. Let $α^2=1/L$, $ε>0$, and define
  the canonical transformation by
  \begin{align}
    \label{eq:min-couple-rescaling}
    x &= α + u, &&& X &= U, &&& y &= w/√ φ, &&& Y &= √ φ W.
  \end{align}
  Then the hamiltonian $H=-H^{(1)}+H^{(2)}$~\eqref{eq:min-couple-H-resolved} is transformed to
\begin{align}
  \label{eq:kam-min-couple-rescaled-H}
  H^{(1)} &= ½ \left[ U ² + 2 u ² + 2 √ L u ³ + ½ L u ⁴ \right] \\\notag
  H^{(2)} &= ½ φ \left[ W ² + ½ w ² + φ √ L u w ² + ½ φ L (uw)² +  v_3 w ³ / \sqrt{φ ³ L} + v_3 uw ³ / \sqrt{φ ³} + v_4 w ⁴/φ ²  \right].
\end{align}
\end{lemma}

The quadratic terms in $H$ can be used to put $H$ into Birkhoff normal form:

\begin{lemma}
  \label{lem:kam-min-couple-bnf2}
  Let
  \begin{align}
    \label{eq:kam-actions}
    I_1    & = \dfrac{1}{√ 2} \left[ ½ U ² + u ² \right],                                                    &        &  & I_2 & = \dfrac{1}{√ 2} \left[ W ² + ½ w ² \right].
  \end{align}
  If $φ ∉ 2\set{0,1/3,1/2,1,2,3}$, then there is a canonical
  transformation $(θ_1,I_1,θ_2,I_2) → (u,U,w,W)$ that transforms $H$
  to
  \begin{align}
    \label{eq:kam-H-bnf}
    H      & = - √ 2 I_1 + \dfrac{φ}{√ 2} I_2 + A_{11}I_1^2 + 2 A_{12}I_1 I_2 + A_{22}I_2^2 + O(5),          & \text{where } \\\notag
    A_{22} & = \dfrac{L ² φ ⁶ (2 φ ² - 3) + 24 L φ ² v_4 (φ ² - 1) - 60 v_3 ² (φ ² - 1)}{16 L φ ⁴ (φ ² -1 )} & \text{and}    \\\notag
    A_{12} & = - \dfrac{L φ}{8} \left( \dfrac{3 φ ² - 2}{φ ² - 1} \right),                                   & A_{11} & = 3L/4.
  \end{align}
\end{lemma}

Let $H=H_2 + O(5)$, where $H_2$ is the second-order (in $I_1, I_2$)
Birkhoff invariant of $H$ at the fixed point. The canonical variables
$(θ_1,I_1,θ_2,I_2)$ constitute a system of (singular) angle-action
variables for $H_2$. The theory of Kolmogorov, Arnol'd and Moser on
the preservation of conditionally periodic motion in perturbations of
integrable, real-analytic hamiltonians implies that if either the
hessian or bordered hessian of $H_2$ is non-singular at $I_1=I_2=0$,
then there is a positive-measure set of invariant tori for $H$ whose
density approaches $1$ as
$I_1^2 + I_2^2 →
0$~\cite{MR0097598,MR0147741,MR0170705,MR0163025,MR1829194}. Moreover,
since an invariant $2$-torus separates a $3$-dimensional iso-energy
surface of $H$, the equilibrium is stable.

\begin{theorem}
  \label{thm:kam-min-couple}
  Assume the hypotheses of lemmas~\ref{lem:kam-min-couple-rescaling}
  and \ref{lem:kam-min-couple-bnf2}. Then, there exists a
  positive-measure set of invariant tori for $H$ that accumulates on
  the elliptic critical point $(x=α=1/√ L, X=0, y=0, Y=0)$.
\end{theorem}
\begin{proof}
  By the remarks preceding the theorem, it suffices that for all $φ$
  in the non-resonant set, i.e. $φ ∉ 2\set{0,1/3,1/2,1,2,3}$, either
  the hessian or bordered hessian of $H_2$ is non-degenerate at
  $I_1=I_2=0$. By inspection, each respective determinant is a
  rational function of $φ$. Let $m$ and $n$ be the numerators of the
  respective determinants, with factors of the form $φ^k$ removed. One
  computes that
  \begin{equation}
    \label{eq:kam-min-couple-resultant}
    \resultant m n {φ} =
    \begin{cases}
      2^{12}\,3^6\,5^6\,L^{24}\,v_{3}^{12} & \text{if } v_3 ≠ 0,          \\
      2^{12}\,3^4\,L^{10}\,v_{4}^4         & \text{if } v_3 = 0, v_4 ≠ 0, \\
      1                                    & \text{if } v_3=v_4=0.
    \end{cases}
  \end{equation}
  Since the resultant is never zero, the determinants do not vanish
  simultaneously at a non-zero value of $φ$.
\end{proof}

\begin{remark}
  There is a second, less pedestrian and computation-free, proof of
  theorem~\ref{thm:kam-min-couple} when $φ$ satisfies a Diophantine
  condition. If the Birkhoff normal form of $H$ is trivial at all
  orders and $φ$ satisfies a Diophantine condition, then by a theorem
  of R{ü}ssmann~\cite{MR213679}, $H$ is conjugate to its
  linearization. This cannot happen since $H$ also possesses a
  hyperbolic critical point at the origin. This implies that for
  Diophantine, and hence for almost all, $φ$, the Birkhoff normal form
  of some order is not trivial and therefore the image of the
  frequency map does not lie in a line through the origin. A second
  work of R{ü}ssman implies that the Hamiltonian is ``non-degenerate
  enough'' that a positive measure set of invariant tori exist in a
  neighbourhood of the critical point~\cite{MR1390625,MR1843664}. This
  idea was used by Churchill, Pecelli, Sacolick \& Rod in their study
  of a Yang-Mills-type hamiltonian~\cite{MR494256}.
\end{remark}

\section{Acknowledgments}
\label{sec:ack}

This research has been partially supported by the Natural Science and
Engineering Research Council of Canada grant 320 852.

\section{Conclusion}
\label{sec:conclusion}

This paper studies a \FLRW{} space-time with a coupled field. In
positively-curved space-times, with either a minimal or minimal
conformal coupled field, the coupling between the field and the radius
of space splits a saddle connection and creates a family of horseshoes
on a nearby energy levels. The $C^1$ structural stability of a
horseshoe implies that for ``nearby'' coupled models (i.e. for
lagrangians $L$~\eqref{eq:frw-lagrangian} with $ξ$ close to $0$ or
$1/6$), the horseshoes persist. This implies real-analytic
non-integrability on a general energy. However, the situation on the
important zero energy level still remains inaccessible using the
current techniques.

The situation is similar \& different for negatively-curved
space-times. Similar, in that the $C^1$ structural stability extends
the results to nearby, non-minimally coupled, models. Different, in
the minimal conformal coupled model, because the existence of
horseshoes is proven only for hamiltonians that have a weak coupling
and because the horseshoe is shown to exist on the zero energy level.

These results suggest many intriguing questions. I pose a few:

\begin{questions}[$k=1$]
\item Can the splitting results be extended from a neighbourhood of
  the saddle-centre to prove the hyperbolic periodic orbits on
  $\set{H=0}$ also have split connections?
\end{questions}
\begin{questions}[$k=-1$]
\item In the minimal conformal coupled model, does the splitting
  results extend from weakly coupled to all coupling strengths?
\item The proof of theorem~\ref{thm:proof-of-conf-couple-k=-1} is the
  only place where the hypothesis that the field $Ψ$ is $\R^n$-valued
  is really needed; does the theorem extend to the case when $Ψ$ takes
  values in a smooth manifold?
\end{questions}
\begin{questions}[$k=± 1$]
\item With reference to either \eqref{eq:maciej-min-couple-H} or
  \eqref{eq:maciej-conf-couple-H}, what happens to the family of
  hyperbolic periodic orbits with large values of angular momentum
  $ω$? Are there values where the connections do not split?
\item Beyond the remarks above, what can be proven for the
  non-minimally coupled models?
\end{questions}



\ltblistoffigures
\bibliographystyle{siam}
\bibliography{cosmo-references}
\end{document}


%% file: definitions.tex
\usepackage{verbatim}
\usepackage{hyperref}
\usepackage{xy}
\xyoption{all}
\usepackage{amssymb,amsmath,mathtools,comment}

\newif\ifingfxfile
\ifcsname geometry\endcsname
\ingfxfiletrue
\else
\ingfxfilefalse
\usepackage[margin=1.5in]{geometry}
\fi

\ifcsname showkeystrue\endcsname
\usepackage[notcite,notref]{showkeys}

\fi

\def\useemaxima{0}
\ifnum\useemaxima=1
\usepackage{emaxima}
\else

\fi
\excludecomment{maximacode}

\usepackage{ltbunicode}
\def\gettimestamp#1<#2>{\def\timestamp{\tt #2}}

\newcommand{\resultant}[3]{{\mathrm{resultant}}(#1,#2; #3)}
\newcommand{\ensuretext}[1]{\ifmmode{\text{#1}}\else{#1}\fi}

\ifcsname D\endcsname\let\newDdef=\renewcommand\else\let\newDdef=\newcommand\fi
\newDdef{\D}[1]{%
  \def\emptyarg{}%
  \edef\argone{#1}%
  \ifx\argone\emptyarg%
  \def\dsp{\message{EMPTY}}%
  \else%
  \message{ARGONE:}
  \def\dsp{\,}%
  \fi%
  {\mathrm d}\dsp{}#1}
\newcommand{\ddt}[3][1]{%
  \bgroup%
  \let\partial=\D\def\,{}%
  \edef\ddtorder{#1}%
  \ifnum\ddtorder=1%
  \frac{\partial{#2}}{\partial{#3}}%
  \else%
  \frac{\partial{}^\ddtorder #2}{\partial{} {#3}^\ddtorder}
  \fi%
  \egroup%
}
\newcommand{\lmapsto}[1][r]{\ar@{<-}[#1]}

\newcommand{\xymat}[2][]{\bgroup\def\inxymatrix{}\xymatrix#1{#2}\egroup}

\newcommand{\R}[0]{{\mathbb R}}

\newcommand{\sphere}[1]{{\mathbb S}^{#1}}

\let\ltbasymp=\asymp
\def\asymp#1{\mathrel{\inasymp{#1}}}
\newcommand{\inasympop}[1]{%
  \setbox1=\hbox{$\ltbasymp{}$}%
  \setbox2=\hbox{\hspace{-3.7mm}${}_{{}_{#1}}$}%
  \hbox{\raise .1mm\box1\lower .1mm\box2}}
\newcommand{\inasymp}[1]{\hspace{.1em}\mathrel{\inasympop{#1}}\hspace{.1em}}

\newcommand{\hgpar}[1]{{\mathrm #1}}
\newcommand{\hgf}[4]{F(\hgpar{#1},\hgpar{#2};\hgpar{#3};#4)}
\newcommand{\legendref}[7]{
  \def\ymplus{+}\def\ymminus{-}
  \def\plusz{+}\def\minusz{-}
  \ifnum #1=0
  \def\ymexp{\frac{#3}{2}}
  \def\ymc{#6}
  \else
  \def\ymexp{-\frac{#3}{2}}
  \def\ymc{\bar{#6}}
  \fi
  \ifnum #2=0
  \def\ymratio{ \frac{1+{#7}}{1-{#7}} }
  \def\ymw{ \frac{1-{#7}}{2} }
  \else
  \def\ymratio{ \frac{1-{#7}}{1+{#7}} }
  \def\ymw{ \frac{1+{#7}}{2} }
  \fi
  \left[ \ymratio \right]^{\ymexp}\, \hgf{#4}{#5}{\ymc}{\ymw}
}

\newcommand{\textem}[1]{{\em #1}}
\newcommand{\defn}[1]{\textem{#1}}

\newcommand{\scalarcurvature}[0]{{\mathfrak{Ric}}}
\newcommand{\normof}[1]{\left| #1 \right|}
\newcommand{\gnormof}[1]{|\hspace{-.1em}| #1 |\hspace{-.1em}|}
\newcommand{\ip}[2]{\langle #1, #2 \rangle}
\newcommand{\set}[1]{ \left\{  #1 \right\}}
\newcommand{\diagmatrix}[2]{
  \begin{bmatrix}
    #1 & &\\
    & \ddots & \\
    && #2
  \end{bmatrix}
}
\newcommand{\unstablemanifold}[2][+]{W^{#1}(#2)}
\newcommand{\stablemanifold}[2][-]{\unstablemanifold[#1]{#2}}
\newcommand{\localunstablemanifold}[2][+]{W^{#1}_{loc}(#2)}
\newcommand{\localstablemanifold}[2][-]{\localunstablemanifold[#1]{#2}}
\newcommand{\poincaremelnikov}{Poincar{é}-Melnikov}
\newcommand{\etal}[0]{{\it et. al.}}
\def\Maciejewskiseen{0}
\newcommand{\Maciejewski}[1][0]{\ifnum#1=\Maciejewskiseen{Maciejewski, Przybylska, Stachowiak \& Szydowski}\else{Maciejewski, \etal{}}\fi\gdef\Maciejewskiseen{1}}
\def\Coelhoseen{0}
\newcommand{\Coelhoetal}[1][0]{\ifnum#1=\Coelhoseen{Coelho, Skea \& Stuchi}\else{Coelho, \etal{}}\fi\gdef\Coelhoseen{1}}
\newcommand{\FLRW}[0]{Friedman-Lema{î}tre-Robertson-Walker}
\newcommand{\rma}[0]{\mathrm{a}}
\newcommand{\rmb}[0]{\mathrm{b}}
\newcommand{\rmc}[0]{\mathrm{c}}

\newcommand{\SOrth}[1]{\mathrm{SO}_{#1}}
\newcommand{\hess}[1]{{\mathrm{hess}\,#1}}
\newcommand{\didi}[3][0]{\ifnum#1=0\frac{∂ #2}{∂ #3}\else\dfrac{∂ #2}{∂ #3}\fi}
\newcommand{\ditdit}[3][0]{\ifnum#1=0{\frac{∂^2 #2}{∂{#3}^2}}\else{\dfrac{∂^2 #2}{∂{#3}^2}}\fi}
\newcommand{\dash}[1]{{#1}'}
\newcommand{\ddash}[1]{{#1}''}

\ifingfxfile\relax\else
\newtheorem{lemma}{Lemma}[section]
\newtheorem{theorem}{Theorem}[section]

\newtheorem{definition}{Definition}[section]
\newtheorem{proposition}{Proposition}[section]
\theoremstyle{remark}
\newtheorem{remark}{Remark}[section]
\makeatletter
\def\ltblistoffigures{%
    \ifcsname efigurewidth\endcsname\relax\else\gdef\efigurewidth{0.7\textwidth}\fi
    \newpage
    \section*{Figures}%
    \@ltblistoffigures\newpage}
\def\@ltblistoffigures{}
\long\def\beginefigure#1 \endefigure{\g@addto@macro{\@ltblistoffigures}{\begin{figure}[h!]#1\end{figure}}}
\makeatother
\newcounter{cthypothesis}
\newcounter{ctquestion}
\makeatletter
\newcommand{\ltbcustomlabel}[2]{\protected@write \@auxout {}{\string \newlabel {#1}{{#2}{\thepage}{#2}{#1}{}} }%
  \hypertarget{#1}{}}
\makeatother
\newenvironment{hypotheses}
{\begin{list}
    {\textbf{H\arabic{enumi}}.}
    {\usecounter{enumi}
      \setlength{\labelsep}{1em}
      \setlength{\topsep}{10pt plus2ptminus2pt}
    }
    \def\label##1{\ltbcustomlabel{##1}{{\textbf{H\arabic{enumi}}}}}%
    \setcounter{enumi}{\value{cthypothesis}}}%
{\setcounter{cthypothesis}{\value{enumi}}
\end{list}}
\newenvironment{althypotheses}
{\begin{list}
    {\textbf{H\arabic{cthypothesis}\Alph{enumi}}.}
    {\usecounter{enumi}
      \setlength{\labelsep}{1em}
      \setlength{\topsep}{10pt plus2ptminus2pt}
    }
    \def\label##1{\ltbcustomlabel{##1}{{\textbf{H\arabic{cthypothesis}\Alph{enumi}}}}}%
    \addtocounter{cthypothesis}{1}}%
{\end{list}}
\newenvironment{questions}[1][]
{\subsection*{#1}
  \begin{list}
    {\textbf{Q\arabic{enumi}.}}
    {\usecounter{enumi}
      \setlength{\labelsep}{1em}
      \setlength{\topsep}{10pt plus2ptminus2pt}
    }
    \def\label##1{\ltbcustomlabel{##1}{{\textbf{Q\arabic{enumi}}}}}%
    \setcounter{enumi}{\value{ctquestion}}}%
{\setcounter{ctquestion}{\value{enumi}}
\end{list}}
\fi

\newenvironment{arxivabstract}[0]{\begin{abstract}}{\end{abstract}}
